%% file: arXiv_version.tex
\documentclass[a4paper]{IEEEtran}
\usepackage{amssymb}
\usepackage{amsmath}
\usepackage{amsthm}
\usepackage{amsfonts}
\usepackage{mathtools}
\usepackage{graphicx}	
\usepackage{tikz}
\usepackage{pgfplots}
\usepackage{circuitikz}
\usepackage{cite}
\usepackage{pst-node}

\theoremstyle{plain}
\newtheorem{lemma}{Lemma}
\newtheorem{proposition}{Proposition}
\newtheorem{theorem}{Theorem}

\theoremstyle{definition}
\newtheorem{definition}{Definition}

\theoremstyle{remark}
\newtheorem{remark}{Remark}

\newcommand{\matP}{\mathcal{P}}
\newcommand{\matB}{\mathcal{B}}
\newcommand{\matE}{\mathcal{E}}
\newcommand{\matV}{\mathcal{V}}
\newcommand{\matRe}{\mathbb{R}}
\newcommand{\matI}{\mathcal{I}}
\newcommand{\si}{\Sigma}

\usepackage{fancyhdr}
\fancypagestyle{empty}{
  \fancyhf{}
  \fancyhead[C]{This work has been submitted to the IEEE for possible publication. Copyright may be transferred without notice, after which this version may no longer be accessible.}     
  
}

\title {Dynamics and Stability of Meshed Multiterminal HVDC Networks}

\author{Santiago Sanchez, Alejandro Garces, Gilbert Bergna, Elisabetta Tedeschi  \thanks{S. Sanchez, G. Bergna and E. Tedeshci (santiago.sanchez, @elkraft.ntnu.no) are with the Department of Electric Power Engineering, Norwegian Universtiy of Science and Technology. Electrical Engineering Building E, 3rd Floor
O.S. Bragstads plass 2a
Gloshaugen, Trondheim, Norway}
\thanks{A. Garces (alejandro.garces@utp.edu.co)is with the Department of Electric Power Engineering, Universidad Tecnol\'ogica de Pereira. Carrera 27 10-02 Barrio Alamos - Risaralda - Colombia - AA: 97 - post code: 660003}}

\begin{document}

\maketitle
\thispagestyle{empty}
\begin{abstract}

This paper investigates the existence of an equilibrium point in multiterminal HVDC (MT-HVDC) grids, assesses its uniqueness and defines conditions to ensure its stability.  An offshore MT-HVDC system including two wind farms is selected as application test case. At first, a generalized dynamic model of the network is proposed, using hypergraph theory. Such model captures the frequency dependence of transmission lines and cables, it is non-linear due to the constant power behavior of the converter terminals using droop regulation, and presents a suitable degree of simplifications of the MMC converters, under given conditions, to allow system level studies over potentially large networks. Based on this model, the existence and uniqueness of the equilibrium point is demonstrated by returning the analysis to a load-flow problem and using the Banach fixed point theorem. Additionally, the stability of the equilibrium is analyzed by obtaining a Lyapunov function by the Krasovskii’s theorem. Computational results obtained for the selected 4 terminal MT-HVDC grid corroborate the requirement for the existence and stability of the equilibrium point. 
\end{abstract}

\begin{IEEEkeywords}
Offshore wind energy, Multiterminal HVDC systems, HVDC, transient stability, equilibrium point, dynamic systems.
\end{IEEEkeywords}		

\section{Introduction}

\subsection{Motivation and state of the art}
\IEEEPARstart{T}{he} increasing penetration of renewable energy in the traditional power system and particularly the massive integration of offshore wind farms, whose installed capacity exceeded 18 GW in 2017, will call for the future deployment of multiple terminal HVDC grids. MT-HVDC increase the efficiency of the power transfer over long distances compared to HVAC solutions and the reliability of the power supply compared to point-to-point interconnection. Moreover, the use of the Voltage Source Converter (VSC) technology, with independent control of active and reactive power at the AC terminals is considered and advantage to support weak power systems 
\cite{chinaHVDC}.\\
The deployment of large interconnected HVDC grids requires specific power system analyses, similar to those traditionally applied to AC power systems, which however need to the adapted to describe the specific components and operation of MT-HVDC grids.\\
Classic stability analysis in power systems requires three main steps: modelling, load flow calculation and dynamic analysis. In the first step, the grid is modelled in order to capture the main dynamics of the real system. Typically, an accurate yet simplified model is required to reduce the computational complexity, particularly in the case of large and complex systems. In the second step, a load flow analysis is performed in order to obtain the equilibrium point. Numerical algorithms are required in this step due to the non-linear nature of the equations. Finally, the dynamic analysis is performed around the equilibrium point. It is important to highlight that normally the existence of the equilibrium point and convergence of the numerical algorithm are taken for granted, despite the non-linear nature of the problem.
This paper addresses this aspect by identifying exact conditions for the existence and uniqueness of the solution. \\
The related problem of existence of the equilibrium in electric systems including constant power loads has been previously studied mostly in dc microgrids and distribution level grid applications  \cite{Pflowmit,koguiman,adhoc}. The equilibrium existence was ensured upon compliance with a certain inequality condition for an electric grid with constant power terminals in \cite{koguiman}. In \cite{adhoc} the stability and power flow of dc microgrids were analyzed considering a worst-case scenario. Such contributions focus on the stability for microgrids, and the same approach can be extended to the study of MT-HVDC system, where, not only constant power loads, but also constant power generators and droop controlled units are present.\\
On the other hand, linearized models are used for stability assessment for MT-HVDC has been specifically addressed in the literature (e.g. \cite{JonAre}), also including the effect of droop controllers \cite{drops}). Transient stability studies based on extensive numerical simulation have also been proposed for MT-HVDC systems  \cite{transient}, as well as power flow analyses (e. g.\cite{classico}). Despite being debated for a long time, the latter still represents an open research subject, as proved, for example, by \cite{pedro2} and \cite{multiport}.\\
However, to the authors’ knowledge, all the studies on MT-HVDC grids presented so far take the existence of an equilibrium point for granted without a formal demonstration.
This paper aims at filling this gap, by proposing a method to assess the existence of an equilibrium point in and MT-HVDC network, exploiting the analogy with the power flow convergence problem. The equilibrium is represented by a set of non-linear algebraic functions that require a successive approximations method for their solution. The Banach fixed point theorem is used in order to guarantee convergence of the power flow and uniqueness of the solution. The stability study of the equilibrium point is also presented with the evaluation of a Lyapunov function obtained by the Krasovskii's method \cite{sastry,khalil}. The paper represents a significant extension of the contributions included in \cite{ourmthvdc,MMCPItune}, including a grid representation based on graph theory, which includes the frequency dependence of transmission line/cable models. Additionally, it presents a more realistic model of the Modular Multilevel Converter together with conditions for its simplifications. And a procedure to find the Lyapunov function with the Krasovskii's method based on the solution of a linear matrix inequality problem.
The rest of the paper is organized as follows:  in section \ref{sec:dynamicmodel} the dynamical model of the MT-HVDC is presented with an accurate model of the HVDC lines.  Next, in section \ref{sec:op} the equilibrium point of the grid is studied. After that, the stability of the grid is analyzed by using the classic Lyapunov theorem and Krasovskii's  method in section \ref{sec:stability}. Finally, numerical results are presented in Section \ref{sec:computa}, followed by conclusions in Section \ref{se:conclu}.

\section{Dynamical model}\label{sec:dynamicmodel}

\subsection{Model of the Modular Multilevel Converter}

\begin{figure}
    \centering
    \includegraphics[scale=0.5]{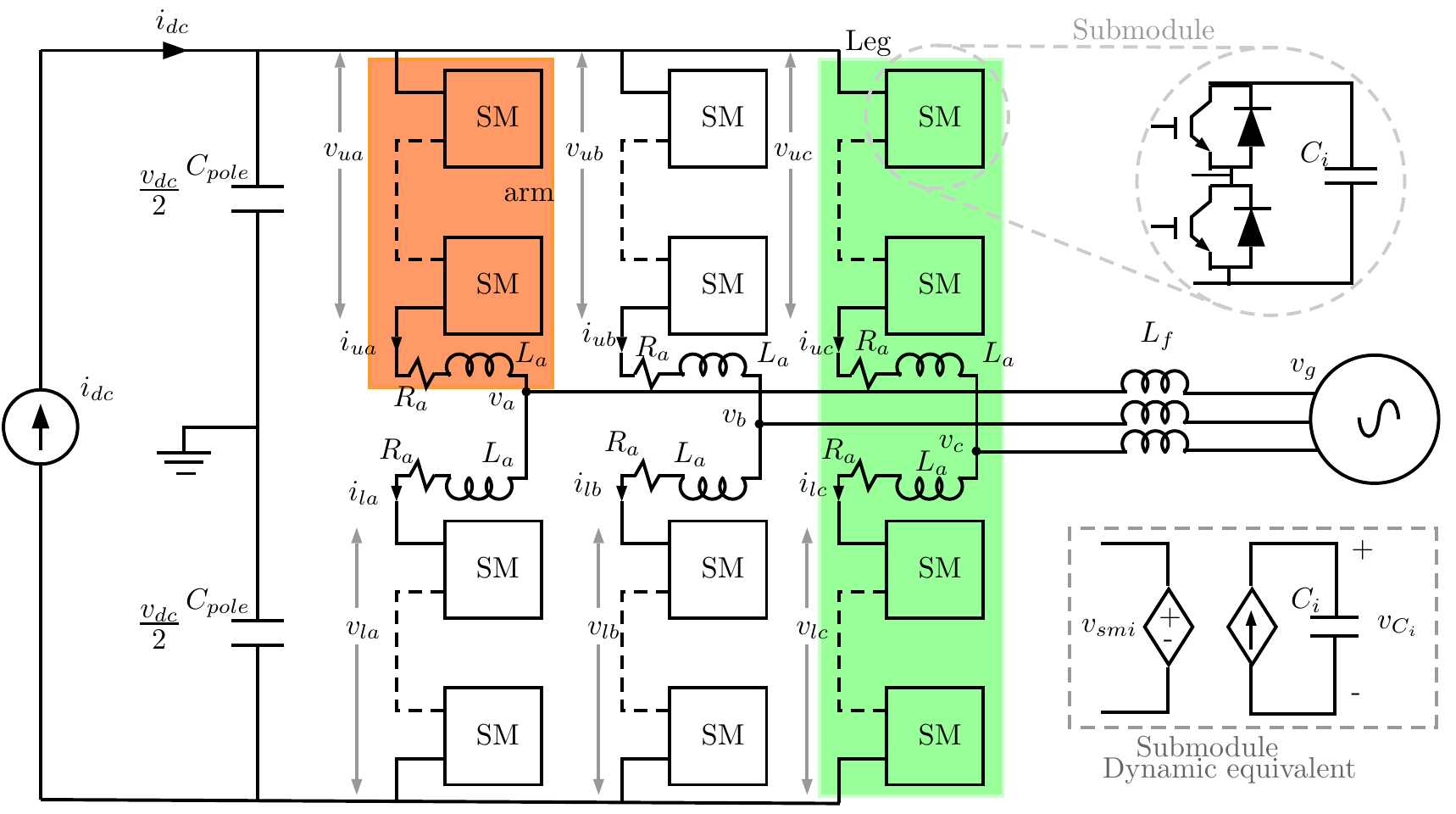}
    \caption{MMC dynamic model}
    \label{fig:mmc}
\end{figure}

The modular multilevel converter is the most promising converter technology for MT-HVDC system in offshore applications. For the purpose of power system studies, its dynamics is described by the continuous model presented in \cite{MMC1,MMC2,elimpublicable}. Initially, the $n$ submodule dynamic equivalents depicted in Fig. \ref{fig:mmc} are aggregated into an equivalent circuit per arm. The equivalent capacitor is defined as $C_{eq}=C_i/n$, where the current flowing into the equivalent capacitor is represented by the Hadamard product  defined with $\circ$:
\[C_{eq}\frac{d v_{eq}}{dt}=m_{}\circ i_{};\]
where, $v_{eq}=(v_{eq,ua},v_{eq,ub},v_{eq,uc},v_{eq,la},v_{eq,lb},v_{eq,lc})^T$, $v_{eq}\in \mathbb{R}^{6\times 1}$ is the equivalent vector of aggregated submodule capacitor voltage at each phase, $i=(i_{ua},i_{ub},i_{uc},i_{la},i_{lb},i_{lc})^T$ is the arm current per phase, the subindices are used for the $\{u,l\}$ upper and lower arms, respectively and the sub-indices $\{a,b,c\}$ represents the phases. The dynamics of the equivalent arm can be controlled by the modulation index $m=(m_{ua},m_{ub},m_{uc},m_{la},m_{lb},m_{lc})^T$. Moreover, the equivalent voltage at the arm $v_{ul}=(v_{ua},v_{ub},v_{uc},v_{la},v_{lb},v_{lc})^T$ is proportional to the equivalent capacitor voltage
\[v_{ul}=m\circ v_{eq}.\]
The model in \eqref{eq:uppv} represents the voltage of the upper arm loops, with $v_{u}=(v_{ua},v_{ub},v_{uc},)^T$, $v_{}=(v_{a},v_{b},v_{c},)^T$, $i_{u}=(i_{ua},i_{ub},i_{uc})^T$ and ${1_x}$ is an all ones vector of size ${3\times 1}$.
\begin{equation}
\label{eq:uppv}
    L_{a}\frac{di_{u}}{dt}=-v-v_{u}+1_x\frac{v_{dc}}{2}-R_{a}i_{u},
\end{equation}
where, $L_a$ and $R_a$ are the arm inductance and resistor, respectively. The dc voltage is $v_{dc}$ and the output voltage is $v$. The dynamics of the lower arm are described by the equation \eqref{eq:lowv}, with $v_{l}=(v_{la},v_{lb},v_{lc})^T$ and $i_{l}=(i_{la},i_{lb},i_{lc})^T$.
\begin{equation}
    \label{eq:lowv}
    L_{a}\frac{di_{l}}{dt}=v_{}-v_{l}+{1_x}\frac{v_{dc}}{2}-R_{a}i_{l}.
\end{equation}
with the addition and subtraction of the upper and lower currents, a new set of differential equations is obtained to represent the dynamics of the converter. This new set of differential equations facilitates the implementation of the control strategy. First, the currents are defined as:
\begin{equation}
    i_{g}=i_{u}-i_{l}
\end{equation}
where $i_{g}=(i_{ga},i_{gb},i_{gc})^T$ is the current flowing into the grid at each phase as depicted in Fig \ref{fig:mmc}. The addition of the upper and lower currents produce a circulating current (defined as $i_{\si}=(i_{\si, a},i_{\si, b},i_{\si, c})^T$).
\begin{equation}
    \label{eq:icirc}
    i_{\si}=\frac{1}{2}(i_{u}-i_{l}).
\end{equation}
Following the procedure above for the voltages, the difference of lower and upper voltages produces a voltage driving the grid current ($e=(e_{a},e_{b},e_{c})^T$), the variable in terms of upper and lower voltages is described in \eqref{eq:ex}.
\begin{equation}
    \label{eq:ex}
    e_{}=\frac{1}{2}(-v_{u}+v_{l})
\end{equation}
The voltage $u_{\si}=(u_{\si,a},u_{\si,b},u_{\si,c})^T$ that drives the circulating current is shown in \eqref{eq:usx}.
\begin{equation}
\label{eq:usx}
u_{\si}=\frac{1}{2}(v_{u}+v_l)
\end{equation}
Therefore, the differential equations that model the system per phase are \eqref{eq:ig} and \eqref{eq:is}.
\begin{eqnarray}
    \label{eq:ig}
    \frac{L_a}{2}\frac{di_{g}}{dt}&=&-\frac{R_a}{2}i_{g}+e-v
    \\
    \label{eq:is}
    L_a \frac{i_{\si}}{dt}&=&-R_a i_{\si}-u_{\si}+{1_x}\frac{v_{dc}}{2}
\end{eqnarray}
The direct and quadrature axes representation of the grid currents are obtained applying the Park transform to the set \eqref{eq:ig} (See the Park's transformation in \cite{MMC2}). The system is represented in vector form with the current grid vector $i_{gdq}^T=\left(i_{gd},i_{gq}\right)$, the voltage vector $v_{dq}^T=\left(v_{d},v_{q}\right)$, the voltage $e_{dq}^T=\left(e_d,e_q\right)$ and the electrical frequency speed is $\omega$.
\begin{eqnarray}
    \label{eq:igdq}
    L_a \frac{di_{gdq}}{dt}=-R_ai_{gdq}+L_aJ_{\omega}i_{gdq}+e_{dq}-v_{dq},
\end{eqnarray}
where, $J_{\omega}$ is a skew-symmetric matrix defined as
\[J=\left(
\begin{array}{cc}
0 &-\omega\\
\omega &0
\end{array}
\right)
\]
The circulating currents at the direct and quadrature axes can be obtained applying the Park transform to a system rotating at $-2\omega$. The vector form of the circulating current for d and q is $i_{\si dq}^T=\left(i_{\si d},i_{\si q}\right)$ and the zero sequence circulating current is $i_{\si z}$. The vector form of the circulating voltage is $u_{\si dq}^T=\left(u_{\si d},u_{\si q}\right)$ and the zero sequence component is $u_{\si z}$.

\begin{eqnarray}
    L_a\frac{di_{\si dq}}{dt}&=&-R_ai_{\si dq}+2L_aJ_{\omega}i_{\si dq}-u_{\si dq},\\
    L_a\frac{di_{\si z}}{dt}&=&-R_ai_{\si z}-u_{\si z}+\frac{v_{dc}}{2}.
\end{eqnarray}

\subsubsection{Currents control strategy}
The control of the MMC is realized by the application of PI controllers to the currents at the grid side, the application of the circulating current suppression control and the zero sequence circulating current with PI too. Without loss of generality a first order current system modeled as \eqref{eq:igen} has a general controller, which is described by \eqref{eq:erru} and \eqref{eq:ugen}.
\begin{eqnarray}
\label{eq:igen}
l \frac{di}{dt}&=&-ri+u+d\\
\label{eq:erru}
\frac{d\phi}{dt}&=&i_{ref}-i\\
\label{eq:ugen}
u&=&k_{p}(i_{ref}-i)+k_i\phi-d
\end{eqnarray}
where, $i$ is the current variable to be controlled, $\phi$ is the controller state variable, $u$ is the control variable, $d$ is the disturbance of the system, $l$ and $r$ are the general variables for the resistor and inductance, respectively. The controller proportional gain is $k_{p}$ and the integral gain is $k_{i}$.

\subsubsection{Zero sequence energy model and control}
It has been shown in \cite{MMC2} that with the appropriate control of the zero\footnote{Without loss of generality the sub-index $z$ is only used in this subsection to represent the zero sequence of electrical systems.} sequence energy of the MMC the system balances the power between the AC grid and the DC grid. Moreover, the variable responsible for the DC power side is $i_{\si z}$. Therefore, this paper uses the control of the zero sequence energy $w_{\si z}$ model.
\begin{equation}
    \label{eq:wz}
    \frac{dw_{\si z}}{dt}\approx 2 u_{\si z} i_{\si z}-\frac{1}{2}(e_{d}i_{gd}+e_q i_{gq})
\end{equation}
The controller calculates the reference current $i_{ref,\si z}$ and is described in \eqref{eq:wu} and \eqref{eq:werr}.
\begin{align}
\label{eq:wu}
i_{ref,\si z}&=\frac{1}{2u_{\si z}}\left(k_{pw}(w_{ref,\si z}-w_{\si z})+k_{iw}\phi_w+P_{d}\right)
\\
\frac{d\phi_{w}}{dt}&=w_{ref,\si z}-w_{\si z}
\label{eq:werr}
\end{align}
where, $w_{ref,\si z}$ is the reference zero sequence energy, $\phi_w$ is the state of the corresponding controller, $P_{d}$ is defined as the disturbance of the system \eqref{eq:wz} and it is $P_{d}=(e_{d}i_{gd}+e_q i_{gq})/2$. The controller gains are the proportional and the integral $k_{wp}$ and $k_{wi}$, respectively.

\textcolor{black}{
\subsection{Conditions for the use of an approximated model for the MMC}
It is a common procedure in power systems analysis to approximate the fast transient behaviour of the power electronics converter to an active power injection model as presented in \cite{milanovic,KOTB201679,lewisugrid,HVDC_EMTP_stab}. This approximation is based on the reasonable assumption that the converters are tightly-regulated, the harmonic distortion is negligible, dc faults are outside the scope of the analysis, the transient response is typically in the range of few milliseconds, losses can be neglected and a generic model can be used to represent the dynamics. The control strategy presented above allows the MMC to act as an active power injection system. It is important to notice that the internal MMC energy is kept constant to control it as an active power source. \\
Therefore, under the assumption that there is perfect regulation of the MMC, the converter model for the network is reduced to \eqref{eq:pinpout}.  
\begin{equation}
    \label{eq:pinpout}
    \tau \frac{di_{pj}}{dt}=-i_{pj}+i_{ref}
\end{equation}
where, $i_{pj}$ is the DC current at the power injection nodes, $\tau$ is the time constant that approximates the current controllers effect, $i_{ref}$ is the input reference of the converter. 
}

\subsection{Model of the grid}

Let us consider an MT-HVDC grid represented by the terminals $\mathcal{N} = \left\{1,2,...,N \right\}$.  Each terminal has a converter with either constant voltage control or constant power with voltage droop control; this is represented by non-empty disjoint sets $\mathcal{V}$ and $\mathcal{P}$ such that $\mathcal{N}=\left\{ \mathcal{V},\mathcal{P}\right\}$. Each transmission cable is characterized by the model depicted in Fig \ref{fig:modelo_linea} which according to \cite{CableJS}, is an accurate frequency dependent model of an HVDC cable.  The model could have several parallel RL elements (The figure depicts a model with only 3 RL branches) and the parallel elements for the capacitance and conductance $C_s$ and $G_{ci}$, respectively.

\begin{figure}
\centering
\footnotesize
\ctikzset{bipoles/length=0.8cm}
\begin{circuitikz}[scale=0.4]
  \draw (0,6) -- (-3,6);
  \draw (8,6) -- (11,6);
  \draw[very thick] (-3,5.5) -- (-3,6.5) node[above] {$n_1$};
  \draw[very thick] (11,5.5) -- (11,6.5) node[above] {$n_2$};
  \draw (0,0) node[ground] {}  to[C,l_=$C_s$] (0,4) -- (2,4) to [L=$L_3$] (4,4) to [R=$R_3$] (6,4) -- (8,4) to [C,l_=$C_s$] (8,0) node[ground] {};
  \draw (-1,0) node[ground]{} to[R=$G_{ci}$](-1,4)-- (0,4);
  \draw (9,0) node[ground]{} to[R,l_=$G_{ci}$](9,4)-- (8,4);
	\draw (0,4) |- (2,6) to [L=$L_2$] (4,6) to [R=$R_2$] (6,6) -| (8,4) ;
	\draw (0,4) |- (2,8) to [L=$L_1$] (4,8) to [R=$R_1$] (6,8) -| (8,4) ;
 \end{circuitikz}
\caption{Improved approximation model with a multi-branch $\pi$ model of and HVDC cable connecting nodes $n1$ and $n2$.}
\label{fig:modelo_linea}
\end{figure}
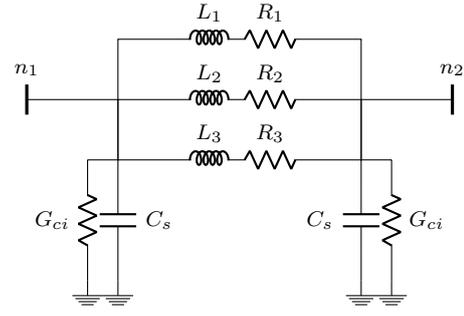

Let us consider $\mathcal{B}$ as the set of the RL sub-branches of each HVDC line.  Therefore, the grid can be represented by a uniform hypergraph $\mathcal{HG} = \left\{\mathcal{N},\mathcal{E}\right\}$. Where $\mathcal{E} \subseteq (\mathcal{N}\times\mathcal{N})\times\mathcal{B}$ is the set of hyper-edges and each one has $\matB$ branches as depicted in Figure \ref{fig:hypergraph}

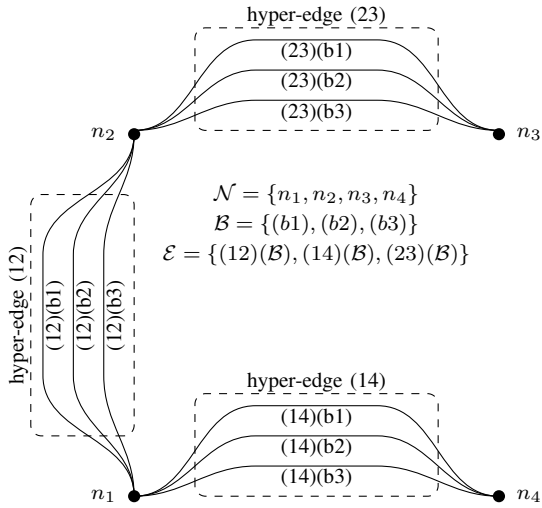
\begin{figure}
\footnotesize
\centering
\begin{tikzpicture}[x=0.8mm,y=0.8mm]
\fill (0,0) circle (1);
\fill (0,60) circle (1);
\fill (60,0) circle (1);
\fill (60,60) circle (1);
\node at (-5,0) {$n_1$};
\node at (-5,60) {$n_2$};
\node at (65,0) {$n_4$};
\node at (65,60) {$n_3$};

\draw (0,0) to[out=0, in=180] +(20,5) -- +(20,0) to[out=0, in=180] +(40,-5);
\draw (0,0) to[out=0, in=180] +(20,10) -- +(20,0) to[out=0, in=180] +(40,-10);
\draw (0,0) to[out=0, in=180] +(20,15) -- +(20,0) to[out=0, in=180] +(40,-15);
\node at (30,3) {(14)(b3)};
\node at (30,8) {(14)(b2)};
\node at (30,13) {(14)(b1)};
\draw[dashed,rounded corners] (10,0) rectangle +(40,17);
\node at (30,19) {hyper-edge (14) };

\begin{scope}[rotate=90]
\draw (0,0) to[out=0, in=180] +(20,5) -- +(20,0) to[out=0, in=180] +(40,-5);
\draw (0,0) to[out=0, in=180] +(20,10) -- +(20,0) to[out=0, in=180] +(40,-10);
\draw (0,0) to[out=0, in=180] +(20,15) -- +(20,0) to[out=0, in=180] +(40,-15);
\node at (30,3)[ rotate=90] {(12)(b3)};
\node at (30,8) [ rotate=90]{(12)(b2)};
\node at (30,13) [ rotate=90]{(12)(b1)};
\draw[dashed,rounded corners] (10,0) rectangle +(40,17);
\node at (30,19) [rotate=90]{hyper-edge (12) };
\end{scope}

\begin{scope}[yshift=138]
\draw (0,0) to[out=0, in=180] +(20,5) -- +(20,0) to[out=0, in=180] +(40,-5);
\draw (0,0) to[out=0, in=180] +(20,10) -- +(20,0) to[out=0, in=180] +(40,-10);
\draw (0,0) to[out=0, in=180] +(20,15) -- +(20,0) to[out=0, in=180] +(40,-15);
\node at (30,3)[ ] {(23)(b3)};
\node at (30,8) []{(23)(b2)};
\node at (30,13) []{(23)(b1)};
\draw[dashed,rounded corners] (10,0) rectangle +(40,17);
\node at (30,19) []{hyper-edge (23) };
\end{scope}

\node at (30,50) {$\mathcal{N} = \left\{ n_1,n_2,n_3,n_4\right\}$};
\node at (30,45) {$\mathcal{B} = \left\{ (b1),(b2),(b3)\right\}$};
\node at (30,40) {$\mathcal{E} = \left\{ (12)(\mathcal{B}),(14)(\mathcal{B}),(23)(\mathcal{B})\right\}$};

\end{tikzpicture}
\caption{Example of a uniform hypergraph wich represents an MT-HVDC with four nodes and three HVDC lines}
\label{fig:hypergraph}
\end{figure}

Let us define the branch-to-node oriented incidence matrix $A = [A_{\matV},A_{\matP}] \in \mathbb{R}^{\matE\times\mathcal{N}}$ as a matrix with $a_{ij} = 1$ if there is an hyper-edge connecting the nodes $i$ and $j$ in the direction $ij$, $a_{ij} = -1$ if there is an hyper-edge connecting the nodes $i$ and $j$ in the direction $ji$, and $a_{ij}=0$ if there is not hyper-edge between $i$ and $j$.  The current in each $RL$ sub-branch requires a triple-sub index which represents the sending node, the receiving node and the sub-branch itself.

Define $1_{\matB}$ is an all-ones vector of size $\matB \times 1$.  Therefore, the currents and voltages are given by

\begin{eqnarray}
V_{\matE}& =& ( 1_\matB\otimes A_{\matV}) \cdot V_{\matV} + (1_\matB\otimes A_{\matP}) \cdot V_{\matP} \\
  I_{\matV} &=& ( 1_\matB\otimes A_{\matV})^{T} \cdot I_{\matE} \\
  I_{\matP} &= &(1_\matB\otimes A_{\matP})^{T} \cdot I_{\matE} 
\end{eqnarray}

where $\otimes$ represents the Kronecker product,  $V_{\matV}\in\matRe^{\matV}$ is the voltage in the voltage controlled terminals, a value given by the tertiary control, $A_{\matP} \in \mathbb{R}^{\matE\times\matP}$ , $A_{\matV} \in \mathbb{R}^{\mathcal{E}\times\mathcal{V}}$ and $I_{\matE}\in\mathbb{R}^{\matE\matB\times  1}$ are the current in each HVDC line, it is defined as $I_{\matE}=(I_{\matE 1}^T,...,I_{\matE \matB}^T)^T$; the sub-indices represent the sub-branch of the hyper-edges.

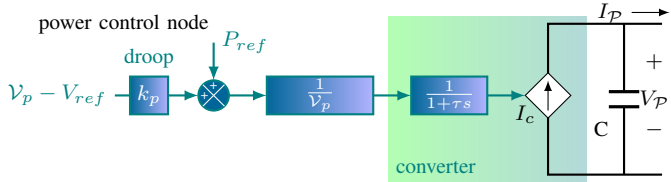
\begin{figure}[tb]
\footnotesize
\centering
\begin{tikzpicture}[x=1mm, y = 1mm]

\begin{scope} [yshift = -70]
\fill[ left color=green!30, right color = green!50!blue!30, thick] (27,-4) rectangle +(26,22);
    \node[right,green!50!blue] at (27,-2) {converter};
\node[right] at (-20,17) {power control node};
	\draw[green!50!blue, left color = blue!60!green, right color = blue!30, thick] (11,5) rectangle +(14,5);
	\node[white] at (18,7.5) {$\frac{1}{\mathcal{V}_{p}}$};
	\draw[green!50!blue, fill = blue!60!green, thick] (4,7.5) circle (2);
	\draw[white] (4,7.5) -- +(45:2);
	\draw[white] (4,7.5) -- +(-45:2);
	\draw[white] (4,7.5) -- +(135:2);
	\draw[white] (4,7.5) -- +(-135:2);
	\node[white] at (4,6.3) {\tiny{\ }};
	\node[white] at (2.8,7.5) {\tiny{+}};
	\node[white] at (4,8.8) {\tiny{+}};
	\draw[green!50!blue, thick,-latex] (-9,7.5) node[left] {$\mathcal{V}_{p}-V_{ref}$} -- +(11,0);
	\draw[green!50!blue, left color = blue!60!green, right color = blue!30, thick] (-7,5) rectangle +(5,5);
	\node[align = center, white] at (-4.5,7.5) {$k_p$};
	\node[green!50!blue] at (-4.5,12.5) {droop};
	
	\draw[green!50!blue, thick,-latex] (4,14.5) node[right] {$P_{ref}$} -- +(0,-5);
	\draw[green!50!blue, thick,-latex] (6,7.5)  -- +(5,0);
	\draw[green!50!blue, thick,-latex] (25,7.5)  -- +(5,0);
	
	\draw[green!50!blue, left color = blue!60!green, right color = blue!30, thick] (30,5) rectangle +(10,5);
	\node[align = center, white] at (35,7.5) {$\frac{1}{1+\tau s}$};
	\draw[green!50!blue, thick,-latex] (40,7.5)  -- +(5,0);
	
	\draw[-, thick] (63,-3) -| + (-15,20) -- +(0,20);
	\draw[-, thick] (58,-3) -- +(0,9);
	\draw[-, thick] (58,17) -- +(0,-9);
	\draw[-, very thick] (56,6) -- +(4,0);
	\draw[-, very thick] (56,8) -- +(4,0);
  \node at (55,3) {C};
	\node at (62,12) {$+$};
	\node at (62,7) {$V_{\matP}$};
	\node at (62,3) {$-$};
	\node at (45,5) {$I_{c}$};
	\draw[-,fill=white] (45,7.5) -- +(3,-3) -- +(6,0) -- +(3,3)  -- cycle;
	\draw[-latex] (48,5.5) -- +(0,4);
	\draw[-latex] (59,18.5) node[left] {$I_{\matP}$} -- +(5,0);
\end{scope}

\end{tikzpicture} 
\caption{Converter model based upon active power injection for a generic MT-HVDC, power control with droop.}
\label{fig:componentes}
\end{figure}
Assuming the converter can be modeled as the active power injection model with a first order system \eqref{eq:pinpout} as shown in Fig. \ref{fig:componentes}, it is possible to write the complete network as:
\begin{align}\label{eq:ie2}
L\frac{d}{dt}I_{\matE}& = -R I_{\matE}+(1_{\matB}\otimes A_{\matV}) V_{\matV} + ( 1_{\matB}\otimes A_{\matP}) V_{\matP},\\
\label{eq:ic2}
\tau \frac{d}{dt}I_{c}&=-I_c+H(V_{\matP}),\\
\label{eq:vp2}
C \frac{d}{dt} V_{\matP} &= I_c - ( 1_{\matB}\otimes A_{\matP})^{T} I_{\matE} {-G_c V_{\matP}},
\end{align}

where, $I_c\in \mathbb{R}^{\mathcal{P}}$ is the converter current vector and $\tau \in \mathbb{R}^{\mathcal{P}\times \mathcal{P}}$ is a diagonal matrix that represents the time constant approximation of the converter current response. $C\in \matRe^{\matP\times\matP}$ is the matrix of the parallel capacitors of the converters and the cable model, $G_c\in \matRe^{\matP\times\matP}$ is a diagonal matrix that contains the shunt conductance of the cable model. 
$L\in \matRe^{\matE\matB\times \matE\matB}$ is the matrix of the inductive parameters of the hyper-edge, it contains the sub-matrices of inductive values per branch, described as:
\begin{equation}
    L=\left(
    \begin{array}{ccc}
         L_{1}&  &\\
         & \ddots &\\
         &&L_{\matB}
    \end{array}
    \right)
\end{equation}
similarly, $R\in \matRe^{\matE\matB\times \matE\matB} $  is the matrix of branch resistances:
\begin{equation}
    R=\left(
    \begin{array}{ccc}
         R_{1}&  &\\
         & \ddots &\\
         &&R_{\matB}
    \end{array}
    \right)
\end{equation}
both, $L$ and $R$ are diagonal and positive definite. Furthermore, $H(V_{\mathcal{P}})$ is a vector function that represents the current on the power terminals as function of its voltage  $V_{\mathcal{P}}$ and the droop control $K_\mathcal{P}$, \textit{i.e.}, 
\begin{equation}
	H(V_{\mathcal{P}}) = diag(V_{\mathcal{P}})^{-1} \cdot \left( S + K_\mathcal{P}\cdot V_{\mathcal{P}}\right),
\end{equation}
\noindent with
\begin{equation}
	S = P_{ref} - K_\mathcal{P}V_{ref}
\end{equation}
\noindent and $P_{ref}$ and $V_{ref}$ the power and voltage references given by the tertiary control. It is assumed that all values are given in per unit.

Let us define the state variables $x = (I^T_{\mathcal{E}},I^T_{c},V^T_{\mathcal{P}})$ then, the non-linear dynamic system that represents the MT-HVDC grid takes the following form

\begin{equation}
    M \dot{x} = \Phi x + h_0(x)+ E_0
    \label{eq:modelo_enx}
\end{equation}

with

\begin{equation*}
M=\left(\begin{array}{ccc}
L_{_{\matB\matE \times \matB\matE}}&0_{_{\matB\matE \times \matP}} &0_{_{\matB\matE \times \matP}}\\
0_{_{\matP \times \matB\matE}}&\tau_{_{\matP \times \matP}} &0_{_{\matP \times \matP}}\\
0_{_{\matP \times \matB\matE}}& 0_{_{\matP \times \matP}}& C_{_{\matP \times \matP}}
\end{array}
\right),
\end{equation*}

\begin{equation*}
\Phi=\left( \begin{array}{ccc}
-R_{_{\matB\matE \times \matB\matE}} &0_{_{\matB\matE \times \matP}}&(1_{\matB}\otimes A_{\matP})_{ _{\matB\matE \times \matP}}\\
0_{_{\matP \times \matB\matE}}& -\matI_{_{\matP \times \matP}}&0_{_{\matP \times \matP}}\\
-(1_{\matB}\otimes A_{\matP})^T_{_{\matP \times \matB\matE}}&\matI_{_{\matP \times \matP}}& -G_{c_{\matP \times \matP}}
\end{array}
\right),
\end{equation*}

\begin{equation*}
h_0(x)=\left(
\begin{array}{c}
0 \\
H(V_{\matP})\\
0
\end{array}
\right),\ \ \ 
E_0=\left(
\begin{array}{c}
(1_{\matB}\otimes A_{\matV} )V_{\matV} \\
0\\
0
\end{array}
\right)
\end{equation*} 
where, $\matI_{\matP\times \matP}\in \matRe^{\matP\times \matP}$ is the identity matrix.

\section{Operating point}\label{sec:op}

Before analyzing the stability of the non-linear system (\ref{eq:modelo_enx}), it is important to establish conditions for the existence of an equilibrium point $x_0$.  

\begin{lemma}[operating point]\label{lemma:1}
An MTDC-HVDC network represented as (\ref{eq:modelo_enx}), admits an equilibrium point $x_0$ with the following representation:
\begin{eqnarray*}
-R\cdot I_{\mathcal{E}} +(1_{\matB}\otimes A_{\mathcal{V}})\cdot V_{\mathcal{V}} + ( 1_{\matB}\otimes A_{\mathcal{P}} )\cdot V_{\mathcal{P}} = 0 \\
-I_c + H(V_{\mathcal{P}}) = 0 \\
H(V_{\mathcal{P}}) - (  1_{\matB}\otimes A_{\mathcal{P}})^{T}\cdot I_{\mathcal{E}} -G_c\cdot V_{\matP}= 0
\end{eqnarray*}
\end{lemma}

\begin{proof}
since it is the equilibrium point, make zero the derivative of $x$ in (\ref{eq:modelo_enx}) and simplify the resulting equations.
\end{proof}

\begin{remark} Notice that the operating point does not depend on the inductances or capacitances of the grid, then, finding the equilibrium point entails the same problem as the power flow.
\end{remark}

\begin{lemma}
The equilibrium point can be completely characterized by a given value of $V_\mathcal{P}$ 
\end{lemma}

\begin{proof}
First, notice that $I_{c}=H(V_\mathcal{P})$ where $H$ is continuous for all $V_\mathcal{P}\neq 0$. Moreover, $R$ is non singular if the graph is connected, and hence, $I_\mathcal{E}$ can be calculated as
\begin{equation*}
    I_\mathcal{E} = R^{-1} (  1_{\mathcal{B}}\otimes A_\mathcal{P}V_\mathcal{P} +  1_{\mathcal{B}}\otimes A_\mathcal{V}V_{\mathcal{V}})
\end{equation*}
therefore we can obtain the equilibrium point from a point $V_\mathcal{P}$ as $x_0=(I_{\mathcal{E}}(V_\mathcal{P}),I_{c}(V_\mathcal{P}),V_\mathcal{P})^T$
\end{proof}

As consequence of this lemma, the equilibrium point can be completely defined by a $V_\mathcal{P}$ that fulfills the following:
\begin{equation*}
\begin{split}
  H(V_\mathcal{P})-
  ( 1_\mathcal{B}^T\otimes A_\mathcal{P}^T )R^{-1}(  1_\mathcal{B}\otimes A_\mathcal{P})V_\mathcal{P} \\ -
  (1_\mathcal{B}^T\otimes A_\mathcal{P}^T )R^{-1}(  1_\mathcal{B}\otimes A_\mathcal{V})V_{\mathcal{V}} - G_\mathcal{P} V_\mathcal{P} = 0 
\end{split}  
\end{equation*}
we can simplify this equation as follows
\begin{equation}
    H(V_{\mathcal{P}}) - \Phi_\mathcal{P} V_{\mathcal{P}} - E_{\mathcal{P}} = 0
\label{eq:equlibrio_enV}
\end{equation}
with
\begin{align*}
    \Phi_\mathcal{P} &= ( 1_\mathcal{B}^T\otimes A_\mathcal{P}^T )R^{-1}( 1_\mathcal{B}\otimes A_\mathcal{P} )V_\mathcal{P} + G_{\mathcal{P}} \\
    E_{\mathcal{P}} &= (1_\mathcal{B}^T\otimes A_\mathcal{P}^T  )R^{-1}(1_\mathcal{B}\otimes A_\mathcal{V} )V_{\mathcal{V}}
\end{align*}
Equation (\ref{eq:equlibrio_enV}) is a non-linear algebraic system which allows multiple roots.  However, only the roots inside a ball $B_0=\left\{V_\mathcal{P}: |v_k-1|\leq \delta \right\}$ have a physical meaning, where $v_k$ can deviate from the nominal voltage. Now, we define conditions for the existence of this root:
\begin{proposition} \label{prop:uno}
Equation (\ref{eq:equlibrio_enV}) admits a root $V_\mathcal{P}$ if there exist a $\delta>0$ and an $\alpha<1$ such that
\begin{equation*}
\alpha = \frac{\left\| \Phi_\matP^{-1}\right\|\left\| S\right\|}{(1-\delta)^2} 
\end{equation*}
this root is unique inside the ball $B_0=\left\{V_\mathcal{P}: |v_k-1|\leq \delta \right\}$ and can be obtained by the successive approximation method.
\end{proposition}

\begin{proof} Define a map $T:B_0\rightarrow B_0$ as 
\begin{equation*}
T(V_{\mathcal{P}}) = \Phi_{\mathcal{P}}^{-1} (H(V_{\mathcal{P}})-E_\mathcal{P})
\end{equation*}
notice that $\Phi_{\mathcal{P}}$ is positive definite and therefore this map is well defined for all $V_{\mathcal{P}}\neq 0$. Evidently, (\ref{eq:equlibrio_enV}) is equivalent to  $  V_{\mathcal{P}} = T(V_{\mathcal{P}}) $ which defines a fixed point. Now consider two points $V_{\mathcal{P}},U_{\mathcal{P}}\in B_0$, then we have that
\begin{equation*}
    \left\| T(V_{\mathcal{P}}) - T(U_{\mathcal{P}})\right\| \leq \left\| \Phi_\matP^{-1} \right\| \left\| (H(V_{\mathcal{P}})-H(U_{\mathcal{P}}))\right\|
\end{equation*}
where $\left\|\cdot\right\|$ is any submultiplicative norm; now,
by using the mean value theorem we can conclude that
\begin{equation*}
    \left\| (H(V_{\mathcal{P}})-H(U_{\mathcal{P}}))\right\| \leq \xi \left\| V_{\mathcal{P}} - U_{\mathcal{P}} \right\|
\end{equation*}
with
\begin{equation}
    \xi =  \underset{V_\mathcal{P}\in B_0}{sup} \left\| \frac{\partial H}{\partial V_\mathcal{P}}\right\|
\end{equation}
in this case, we have that
\begin{equation*}
    \xi = \frac{\left\| S\right\|}{(1-\delta)^2}
\end{equation*}
Therefore
\begin{equation*}
    \left\| T(V_{\mathcal{P}})-T(V_{\mathcal{P}})\right\| \leq \alpha \left\|V_{\mathcal{P}}-U_{\mathcal{P}} \right\|
\end{equation*}
with 
\begin{equation*}
    \alpha = \frac{\left\| \Phi_\matP^{-1}\right\|\left\| S\right\|}{(1-\delta)^2}
\end{equation*}

By using the Banach fixed point theorem \cite{sholomo} we conclude that if $\alpha<1$, then $T$ is a contraction map and there exists a unique fixed point in $B_0$ which can be easily calculated by the method of successive approximations.
\end{proof}
The successive approximation is just the use of a Picard iteration starting from any point $V\in B_0$ (in this case, $v_k=1$ pu as usual in power flow applications):
\begin{equation*}
    V_{\mathcal{P}}^{({iter}+1)} = T\left(V_{\mathcal{P}}^{({iter})}\right)
\end{equation*}
\begin{remark}
Notice that Proposition \ref{prop:uno} does not only guarantee the existence of the equilibrium, but also its uniqueness.  In addition, it gives a numerical methodology to find it.
\end{remark}

\section{Stability analysis}\label{sec:stability}

After finding the equilibrium point our next step is to obtain stability conditions for the MT-HVDC grid by using the Lyapunov theory. In order to simplify our analysis let us define the following: 

\begin{definition}[Incremental model]  The incremental model for a MT-HVDC grid is given by
\begin{equation}
\dot{z}
=
M^{-1}\left(\Phi z+h(z)\right)
\label{eq:incrementalz}
\end{equation}

\noindent where $z = x - {x}_0$ (with $x_0$ the equilibrium point), $x^T=(I^T_{\matE}, I^T_c, V^T_{\matP})$, the operating point is $x_0^T=(I_{\matE 0}^T, I^T_{c 0}, V^T_{\matP 0})$ and $h(z)$ defined as follows 

\begin{equation}
h(z)=\left(\begin{array}{c}
0_{_{\matB\matE \times 1}}\\
H(z_{V_{\matP}})_{_{\matP \times 1}}-I_{c0}\\
0_{_{\matP \times 1}}
\end{array}
\right),
\label{eq:ecuaciongz}
\end{equation}
\end{definition}  

where, the sub-indices represent the dimension of the matrices. The function $H(z_{V_{\matP}})$ with the incremental voltage variable $z_{V_{\matP}}$ is described in \eqref{eq:deltaG1}.
\begin{equation}
\label{eq:deltaG1}
\ H(z_{V_{\matP}})=diag\left(\frac{1}{z_{V_{\matP}}+V_{\matP0}}\right)S{+K_\matP\cdot 1_{\matP}}
\end{equation} 

Let us analyze the stability of this incremental model by using Lyapunov theory.

\begin{theorem}[Lyapunov]\label{theo:Lyap}
Let $E$ be an open subset of $\mathbb{R}^n$ containing $x_0$.  Suppose that $f\in C^1(E)$ and that $f(x_0) = 0$.  Suppose further that there exists a real value function $W\in C^1(E)$ satisfying $W(x_0)=0$ and $W(x)>0$ if $x\neq x_0$.  Then a) if $\dot{W}(x)\leq 0$ for all $x \in E$, then $x_0$ is a stable equilibrium point. b) if $\dot{W}(x)< 0$ for all $x \in E - \left\{x_0\right\}$ is asymptotically stable and c) if $\dot{W}(x)>0$ for all $x\in E - \left\{x_0\right\}$ is unstable.
\end{theorem}

\begin{proof}
See \cite{perko}
\end{proof}

The function $W$ in Theorem \ref{theo:Lyap} is called Lyapunov function.  There is not general method to obtain these type of functions.  Here, we find Lyapunov function candidate with the Krasovskii's method \cite{khalil} in which $W(z)=\dot{z}^T Q(z) \dot{z}$.

\begin{lemma}\label{lemaKrasovskii}

Consider the system $\dot{z}=f(z)$ with f(0)=0. Assume that $f(z)$ is continuously differentiable and its Jacobian $[\partial f/\partial z]$, the Lyapunov function can be differentiated to give\\
\[
\dot{W}(z)=\dot{z}^T\left\{\left(\frac{\partial f}{\partial z}\right)^TQ+Q\left(\frac{\partial f}{\partial z}\right)\right\}\dot{z}=\dot{z}^T\Psi\dot{z}.
\] 
If $Q$ is a constant symmetric, positive definite matrix and $\dot{W}(z)\leq 0$, then the zero solution $z\equiv 0$ is a unique asymptotically stable equilibrium with Lyapunov function $W(z)$.\\
\end{lemma}
\begin{proof}
See \cite{slotine}
\end{proof}

\subsection{Stability for the MT-HVDC system}

\begin{lemma}
The equilibrium of the MT-HVDC network described by \eqref{eq:ie2}, \eqref{eq:ic2} and \eqref{eq:vp2} is asymptotically stable if there exist a matrix $\Psi\preceq 0$, where $\preceq$ represents the L\"{o}wner partial order on negative semidefinite matrices.  

\end{lemma}

\begin{proof}

In order to apply the Lyapunov theorem with the Krasovskii's method to study the stability of the system, the operating point is shifted to the origin as above.  The model is the same for \eqref{eq:incrementalz} with $M$, $\Phi$ and $h(z)$ as in \eqref{eq:modelo_enx} and \eqref{eq:ecuaciongz}.

The Jacobian for the system is given by 

\begin{equation*}
\partial f/\partial z = M^{-1}(\Phi + \partial h(z)/\partial z)
\end{equation*}
that is 
\begin{equation}
\frac{\partial f}{\partial z}=M^{-1}
\left(
\begin{array}{ccc}
-R_{_{\matB\matE \times \matB\matE}}& 0_{_{\matB\matE \times \matP}}& \kappa _{_{\matB\matE \times \matP}}\\
0_{_{\matP \times \matB\matE}}&-\matI_{_{\matP \times \matP}}&
\left(\frac{\partial H(z_{V_{\matP}})}{\partial z_{V_{\matP}}}\right)_{_{\matP \times \matP}}\\
-\kappa^T_{_{\matP \times \matB\matE}}& \matI_{_{\matP \times \matP}}&
-G_{c_{\matP \times \matP}}
\end{array}
\right),
\end{equation} 
Where, the Jacobian matrix $\partial H(z_{V_{\matP}})/\partial z_{V_{\matP}}$ is a diagonal matrix described in \eqref{eq:diagdG};
\begin{equation}
\label{eq:diagdG}
\frac{\partial H(z_{V_{\matP}})}{\partial z_{V_{\matP}}}=diag\left(-\frac{1}{(z_{V_{\matP}}+V_{\matP 0})^2}\right)\cdot diag(S),
\end{equation}
and $\kappa=(1_{\matB}\otimes A_{\matP})$. The constant, symmetric, positive definite matrix $Q \in \matRe ^{(\matB\matE+2\matP)\times(\matB\matE+2\matP)}$ can be calculated by solving the following linear matrix inequality (LMI):
\begin{eqnarray}
\label{eq:LMIQ}
\left\{
\begin{array}{l}
Q=Q^T\succ 0\\
\left(\frac{\partial f}{\partial z}\right)_{j}^TQ+Q\left(\frac{\partial f}{\partial z}\right)_{j}\preceq 0, \forall j \in \{1,2,..r\}.
\end{array}
\right. 
\end{eqnarray}
where the subindex $j=1,2,...,r$ is a set of boundary points of the ball $B_0=\{z_{v_{\matP}}\in \matRe^{\matP}: z_{v_{\matP}}\leq \delta_a\}$ (evidently $z=0$ is inside of this set). 
The linear matrix inequality is feasible if and only if the Jacobian $\left(\partial f/\partial z\right)_j$ is Hurtwitz. Moreover, the asymptotic stability of the system is proved if the LMI \eqref{eq:LMIQ} is feasible. Then, the matrix $Q$ can be calculated and the Lyapunov function is $W(z)=\dot{z}^T Q(z) \dot{z}$. Finally, we invoke \textit{LaSalle's} principle \cite{sastry}. Therefore, $\dot{W}(0)$ is zero \textit{if} $\dot{z}=0$ and it implies $z=0$ from \eqref{eq:incrementalz}.
\end{proof}


\color{black}

\section{Computational results}\label{sec:computa}
The MT-HVDC grid shown in Fig. \ref{fig:MTDCexample} is used to corroborate the analysis presented in the sections above. The cable model is represented with three edges for the series impedance in order to take into account the frequency dependence. Two offshore wind farms and two inland stations were considered. The reference voltages are 1.0 pu. The cable distances as well as the generated and consumed power are described in Fig. \ref{fig:MTDCexample}. The system is represented in per unit with a $P_{base}=400$ MW, $V_{base (dc)}=400$ kV. Additionally, the parameters of the network and the MMCs are described in the appendix. 
For this case, the results of applying Proposition 1 are as follows:
\begin{align}
    \nonumber
     \alpha= 0.0068<1\\
     \nonumber
     \delta= 0.5>0
\end{align}

\begin{figure}[tb]
\footnotesize
\centering
\input{red1.tex}
\caption{Four nodes MT-HVDC grid with two offshore wind farms.}
\label{fig:MTDCexample}
\end{figure}
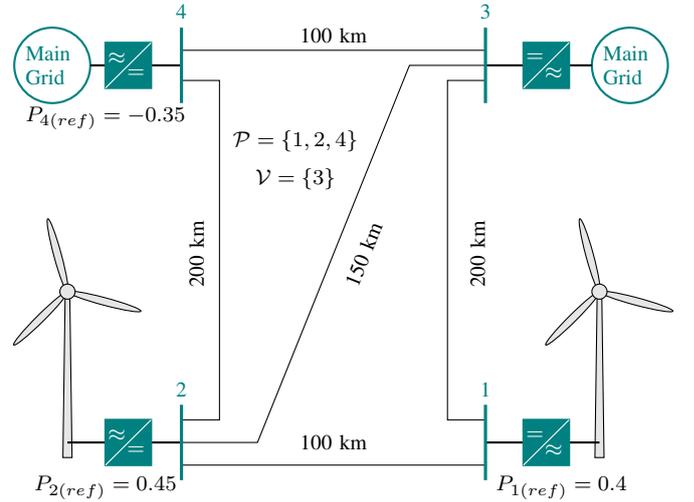
The successive approximations method was evaluated in this system by executing only four iterations. Results are shown in Fig. \ref{fig:superconvergence} were convergence is evident with a logarithmic scale at the y$-$axis. The voltage of the power nodes is $V_{\matP}=(V_{1},V_{2},V_{4})^T= (1.0014,1.00010,0.9998)$ pu and $V_{\matV}=1.0000$ pu.\\

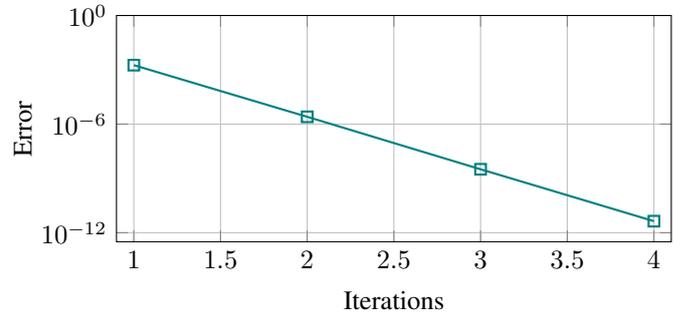
\begin{figure}[tb]
\begin{tikzpicture}
\begin{axis}[scale only axis,width=7.3cm, height=3.0cm,ymode=log, xmajorgrids,ymajorgrids, yminorgrids,minor y tick num=5,
xlabel={Iterations}, ylabel ={Error}, ymax = 1, ymin=0, xmin=0.9, xmax = 4.1]
\addplot [thick, blue!50!green, mark = square]
coordinates{(1,0.0018)(2,2.527E-6)(3,3.1955E-9) (4,4.3337E-12)};
\end{axis}
\end{tikzpicture}
\caption{Convergence of the successive approximations method applied to the MT-HVDC grid from an initial point $V_\mathcal{P}=(1.00,1.00,1.00)^T$}
\label{fig:superconvergence}
\end{figure}

In order to study the stability of the MT-HVDC network, we apply the analysis described in section \ref{sec:stability}, with the evaluation of $\left(\partial f/\partial z\right)_j$ at $2\matP=6$ points of the ball with radius $\delta_a=\delta$. The points are located at the intersections between the axes and the ball as defined above. First, the eigenvalues of $(\partial f/\partial z)_{j}$ have negative real part. The LMI of \eqref{eq:LMIQ} is evaluated with a simple procedure in existing software in MATLAB \cite{lmimatlab} or \cite{cvx}. Once, the LMI problem is feasible, the matrix $Q$ is obtained and $\Psi$ is presented for $z=0$. Figures \ref{fig:Qmatrix} and \ref{fig:Psimatrix} show the matrices values and distribution by a color map representation. It is observed from the color-map that the matrices $Q$ and $\Psi$ are symmetric. All the eigenvalues of $Q$ are greater than zero and all the eigenvalues of $\Psi$ are real and lower than zero.
The eigenvalues of $Q$ and $\Psi$ are $\lambda_{Q}, \lambda_{\Psi} \in \matRe^{(\matB\matE+2\matP)\times(\matB\matE+2\matP)}$, respectively. Therefore, for the operating point presented in the results above: 
$\lambda_{Q}=$\begin{small}
$10^{-4}\times(5,  7,    10,    29,    29,   60,    85,
 92,    313,    313,    1449,   2115,    3443,    3982,$\\
$ 4355,    4870,    6493,   9474,    32646,    40632,    54295)$. 
\end{small}\\
The matrix $\Psi$ has  the following eigenvalues: 
$-\lambda_{\Psi}=$
\begin{small}$(0.6678,
0.8096, 4.6681, 1.5217, 4.2254, 3.7405, 3.2104, 3.2122,$\\
$ 3.2420, 3.2417,3.0355, 2.9865,  2.9131,  2.9169,   2.9255,  2.9563,$\\
$2.9562,2.9481,   2.9437,   2.9386,   2.9410 )$.
\end{small}

\begin{figure}[tb]
    \centering
    \includegraphics[scale=0.8]{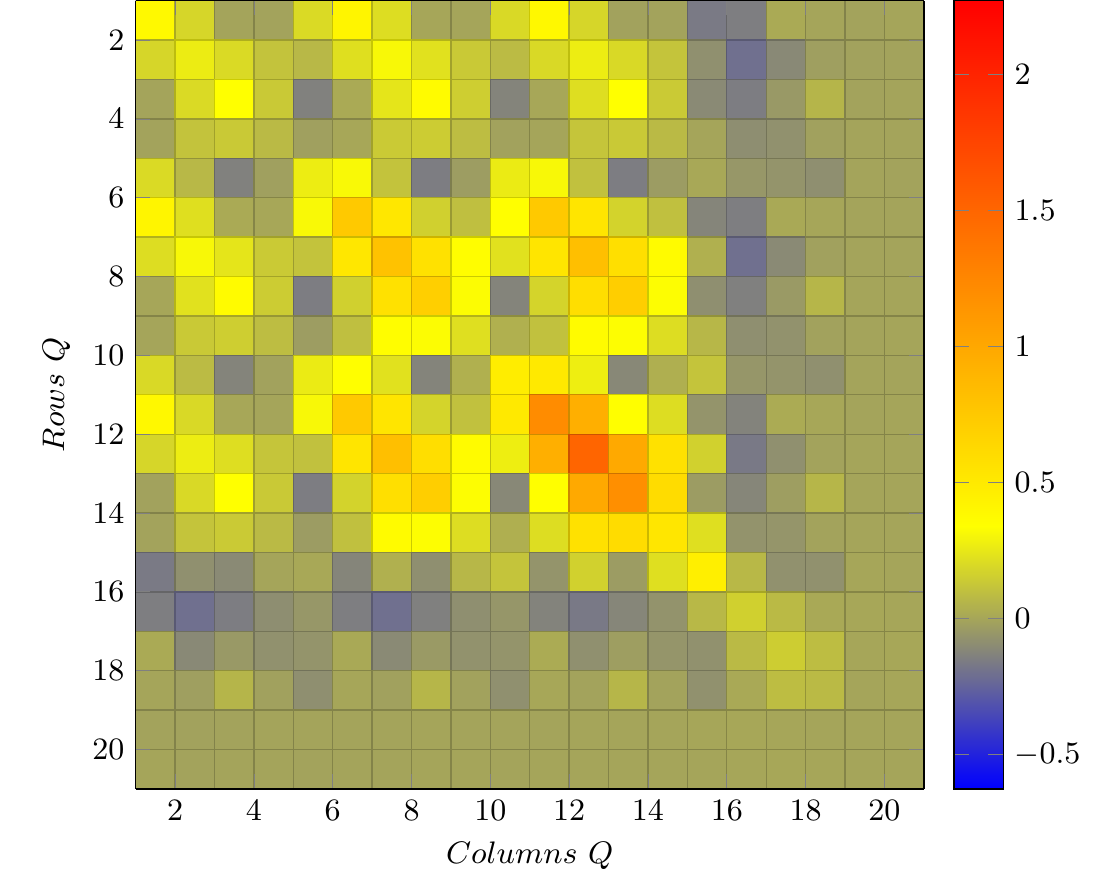}
    \caption{Color-map representation of the matrix Q obtained by the results of the LMI problem.}
    \label{fig:Qmatrix}
\end{figure}
\begin{figure}[tb]
    \centering
    \includegraphics[scale=0.8]{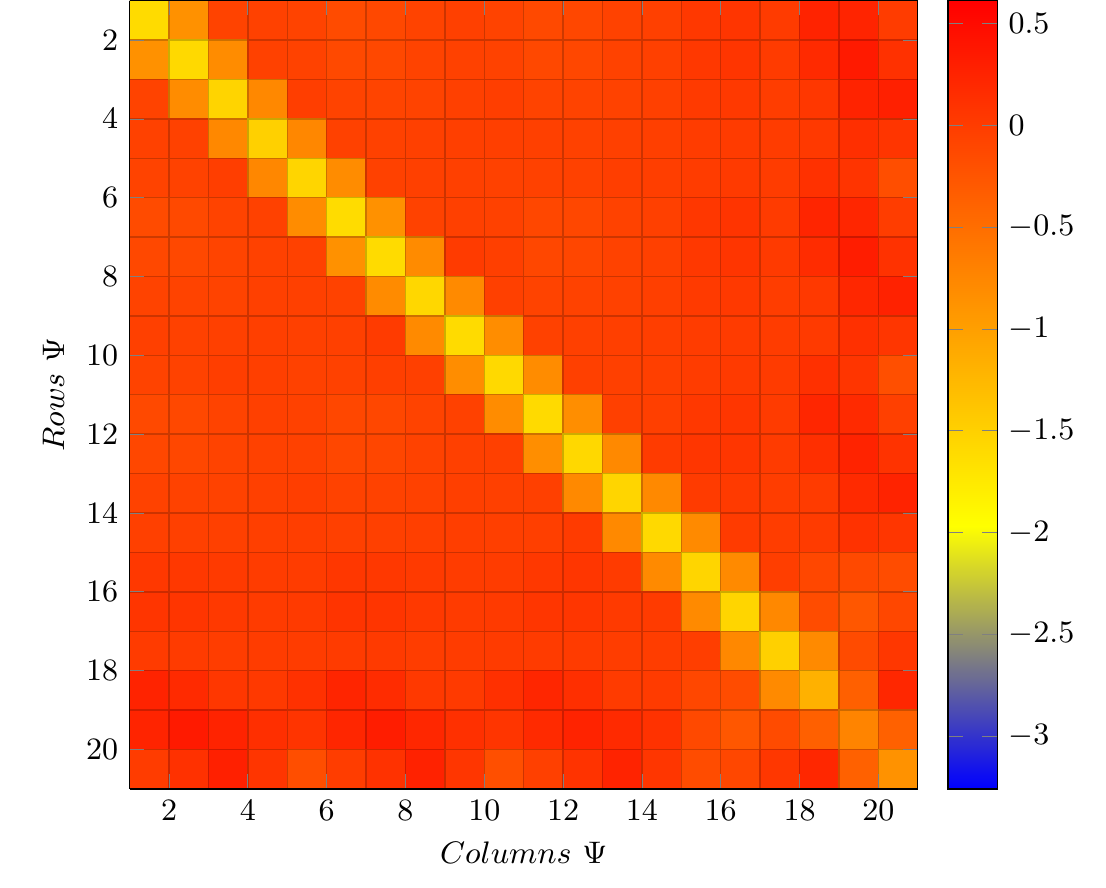}
        \caption{Color-map representation of the matrix $\Psi$ obtained by the results of the LMI problem.}
    \label{fig:Psimatrix}
\end{figure}

The time response results for the four nodes MT-HVDC network are shown in Fig. \ref{fig:3nodesresponse}. The coordination from the black-start condition of the grid uses the closing of the lines at the following instants: initially the line 1-3 and line 1-2 are connected. The line 3-2 is connected at 5.3 s, the next line to be connected is line 2-4 at 6.5 s, finally the line 3-4 is connected at 7 s. The power references for the converters are activated at 4.5 s for converter at node 1, at 5.5 s for the converter at node 2 and finally at 7.7 for the converter at node 3. After the 7.7 s the MT-HVDC grid reaches the desired operating point.

\begin{figure}
    \centering
    \includegraphics[scale=0.8]{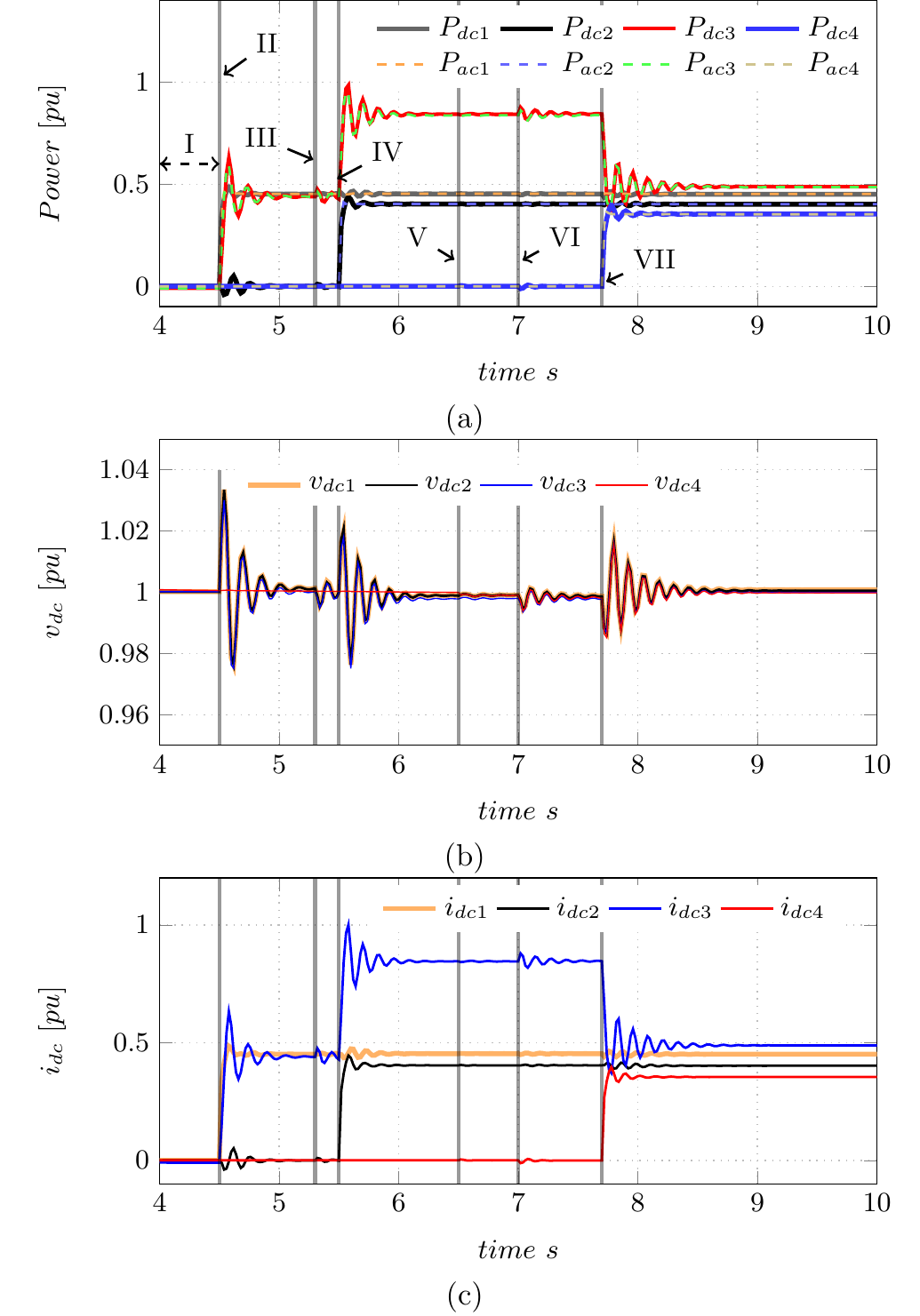}
    \caption{MT-HVDC grid time response with four nodes}
    \label{fig:3nodesresponse}
\end{figure}

\section{Conclusions}\label{se:conclu}
The contribution of the paper folds in three parts, first a generalized model of the MT-HVDC has been described by the use of graph theory and the application of hyper-edges to describe the frequency dependence of the cable model. Therefore, the convergence of the nonlinear algebraic system of a MT-HVDC that describes the operating point and defines the existence of this point has been studied. Moreover, conditions for the existence of the solution of the model of the MT-HVDC have been listed in this paper. The third part of the contribution of the paper is the application of the Krasovskii's method to obtain a Lyapunov function and we list conditions for the stability of the system with an approximated model for the converters as is typically used. Finally, the simulation results present the behaviour of the power flow and transient in a standard network.

\vspace{-0.5cm}
\section{Appendix}
\subsection{Parameters of the dc cable}
The cable model in the network is based on \cite{Freytes2016b}, the parameters are listed in Table \ref{tab:cableparameters}. The parameters $R_{ji}$ and $L_{ji}$ define the cable resistance and inductance at the $j-$th branch, $G_{l}$ and $C_{l}$ are the conductance and capacitance in the cable, respectively.  

\begin{table}
\caption{Parameters of the dc cable.}
\centering
\begin{tabular}{cccc}
Parameter&value&Paramter&value\\
\hline
\ &\ &\ &\ \\
$R_{11}$&1.1724$\cdot 10^{-1}$ [$\Omega$/km]&$L_{11}$&2.2851$\cdot 10^{-4}$ [H/km]\\
$R_{12}$&8.2072$\cdot 10^{-2}$ [$\Omega$/km]&$L_{12}$&1.5522$\cdot 10^{-3}$ [H/km]\\
$R_{13}$&1.1946$\cdot 10^{-2}$ [$\Omega$/km]&$L_{13}$&3.2942$\cdot 10^{-3}$ [H/km]\\
$G_{l}$&7.6330$\cdot 10^{-11}$ [S/km] &$C_{l}$&0.1983$\cdot 10^{-6}$ [F/km]\\
\hline
\end{tabular}
\label{tab:cableparameters}
\end{table}

\subsection{Parameters for the MMC converter}
The parameters of the converters are based on the CIGRE guide \cite{cigrehvdc}. The rated power for each converter is 400 MVA, the system uses as base voltage $220$ kV, the DC base voltage is $400$ kV. the parameters are listed in Table \ref{tab:mmcpar}. The converter transformer is simulated with $L_f$ and $R_f$ as shown in Fig. \ref{fig:mmc}. The inductance of the transformer is $L_f=0.18$ pu, the resistance is $R_{f}=0.6/100$ pu.
The controller parameters are calculated with the procedure presented in \cite{MMCPItune} with the pole placement technique. The parameter $\tau$ for the converter is calculated with the closed loop step response circulating current control. Hence, $\tau\approx 1.2$ ms.

\begin{table}
\caption{Parameters of the MMC.}
\centering
\begin{tabular}{cccccc}
Parameter&value&Parameter&value&Parameter&value\\
\hline
$R_{a}$&0.003 pu &
$L_a$&0.15 pu&
$n$&200 \\
$C_{i}$&5 mF &
$C_{pole}$&60 ms &&\\
\hline
\end{tabular}
\label{tab:mmcpar}
\end{table}
\vspace{-0.3cm}

\bibliographystyle{IEEEtran}
\bibliography{biblio}

\end{document}

%% file: red1.tex
\begin{tikzpicture}[x=1mm,y=1mm]
		\node(T1) at (-5,20) {};
		\draw[rotate = -15, fill=gray!20](T1)+(5,0) ellipse (5 and 0.5);
		\draw[rotate = 105, fill=gray!20](T1)+(5,0) ellipse (5 and 0.5);
		\draw[rotate = 225, fill=gray!20](T1)+(5,0) ellipse (5 and 0.5);
		\draw[-, fill=gray!20] (T1)+(-0.5,0) -- +(-0.3,0) -- +(-0.6,-22) -- +(0.6,-22) -- +(0.3,0);
		\draw[-, fill=gray!20] (T1) circle (1);				
		\node(T2) at (65,20) {};
		\draw[rotate = -15, fill=gray!20](T2)+(5,0) ellipse (5 and 0.5);
		\draw[rotate = 105, fill=gray!20](T2)+(5,0) ellipse (5 and 0.5);
		\draw[rotate = 225, fill=gray!20](T2)+(5,0) ellipse (5 and 0.5);
		\draw[-, fill=gray!20] (T2)+(-0.5,0) -- +(-0.3,0) -- +(-0.6,-22) -- +(0.6,-22) -- +(0.3,0);
		\draw[-, fill=gray!20] (T2) circle (1);		
\draw[thick, blue!50!green] (-7,50) circle (5);	
\draw[thick, blue!50!green] (69,50) circle (5);		
\node at (-7,50)[text width = 20, blue!50!green] {Main Grid};
\node at (69,50)[text width = 20, blue!50!green] {Main Grid};
\draw[thick] (-5,0) --+(15,0);
\draw[thick] (65,0) --+(-15,0);
\draw[thick] (-2,50) --+(12,0);
\draw[thick] (64,50) --+(-14,0);

\draw[] (10,-3) -- (50,-3);		
\draw[] (10,3)  -- +(5,0) -- (15,48) -- (10,48);		
\draw[] (50,50)  -- +(-10,0) -- (20,0) -- (10,0);	
\draw[] (50,3)  -- +(-5,0) -- (45,48) -- (50,48);		
\draw[] (10,52) -- (50,52);		
\draw[thick,blue!50!green, fill] (0,-3) rectangle +(6,6);
\draw[white] (0,-3) -- +(6,6);
\node[white] at (1.5,1) {$\approx$};
\node[white] at (4,-1) {$=$};

\draw[thick,blue!50!green, fill] (55,-3) rectangle +(6,6);
\draw[white] (55,-3) -- +(6,6);
\node[white] at (56.5,1) {$=$};
\node[white] at (59,-1) {$\approx$};

\draw[thick,blue!50!green, fill] (0,47) rectangle +(6,6);
\draw[white] (0,47) -- +(6,6);
\node[white] at (1.5,51) {$\approx$};
\node[white] at (4,49) {$=$};

\draw[thick,blue!50!green, fill] (55,47) rectangle +(6,6);
\draw[white] (55,47) -- +(6,6);
\node[white] at (56.5,51) {$=$};
\node[white] at (59,49) {$\approx$};

\draw[very thick,blue!50!green] (10,-5) -- +(0,10) node[above] {2};
\draw[very thick,blue!50!green] (50,-5) -- +(0,10) node[above] {1};
\draw[very thick,blue!50!green] (10,45) -- +(0,10) node[above] {4};
\draw[very thick,blue!50!green] (50,45) -- +(0,10) node[above] {3};

\node at (25,40) {$\mathcal{P} = \left\{ 1,2,4\right\}$};
\node at (25,35) {$\mathcal{V} = \left\{ 3\right\}$};
\node at (30,0) {100 km};
\node at (30,54) {100 km};
\node [rotate=90] at (12,25) {200 km};
\node [rotate=90] at (49,25) {200 km};
\node [rotate=68] at (34,25) {150 km};
\node at (0,-6) {$P_{2(ref)} = 0.45$};
\node at (60,-6) {$P_{1(ref)} = 0.4$};
\node at (0,43) {$P_{4(ref)} = -0.35$};

\end{tikzpicture}